\documentclass[a4paper,12pt]{article}
\title{The Arcsine law and an asymptotic behavior of orthogonal polynomials}
\author{Hayato Saigo\footnote{E-mail: h\_saigoh@nagahama-i-bio.ac.jp } 
\\ Nagahama Institute of Bio-Science and Technology \\  Nagahama 526-0829, Japan
\\ and
\\ Hiroki Sako\footnote{Email: sako@ie.niigata-u.ac.jp}
\\ Niigata University, Institute of Science and Technology,\\ Niigata 950-2181, Japan
}
\date{}

\usepackage{amsmath}
\usepackage{amsthm}
\usepackage{amssymb}
\usepackage{amsbsy}
\usepackage{amsfonts}
\usepackage{amstext}
\usepackage{amscd}
\usepackage[dvips]{epsfig}

\numberwithin{equation}{section}
\theoremstyle{plain}
\newtheorem{thm}{Theorem}[section]

\newtheorem{cor}[thm]{Corollary}
\newtheorem{lem}[thm]{Lemma}
\theoremstyle{definition}
\newtheorem{exa}[thm]{Example}

\newtheorem{rem}[thm]{Remark}
\newtheorem{df}[thm]{Definition}
\newtheorem{nota}[thm]{Notation}

\begin{document}

\maketitle
\begin{abstract}
Interacting Fock space connects the study of quantum probability theory, classical random variables, and orthogonal polynomials. 
It is a pre-Hilbert space associated with creation, preservation, and annihilation processes.
We prove that if three processes are asymptotically commutative,
the arcsine law arises as the ``large quantum number limits.'' 
As a corollary, it is shown that for many probability measures, asymptotic behavior of orthogonal polynomials is described by the arcsine function.
A weaker form of asymptotic commutativity provides us a discretized arcsine law.

\end{abstract}

\section{Introduction}
The distribution $\mu_{As}$ defined as
\[
\mu_{As} (dx)= \frac{dx}{\pi \sqrt{2-x^2}}, \quad \left(-\sqrt{2}<x<\sqrt{2}\right).
\]
is called the arcsine law. 
In \cite[Theorem 3.1]{SAI} we have proved 
that the arcsine law appears as the large number limit of quantum harmonic oscillator, in the framework of 
algebraic probability theory (also known as ``noncommutative probability theory'' or ``quantum probability theory''). 

The purpose of this paper is to extend this phenomenon in general interacting Fock spaces defined in \cite{A-B}. The interacting Fock space is a triplet of annihilation $A$, preservation $B$, creation $C$ processes in a pre-Hilbert space. The summation $X = A + B + C$ provides us an algebraic random variable.
A usual random variable with finite moments is identical to some algebraic random variable $X$ in law.
Based on this, we often call the equality $X = A + B + C$ a quantum decomposition.

For the random variable $x$ on the Gaussian probability space $\left(\mathbb{R}, \dfrac{\exp(-x^2 /2) dx}{\sqrt{2\pi}}\right)$, we obtain a quantum decomposition $X = A + C$, ($B = 0$) with the canonical commutation relation $[C, A] = id$. 
The interacting Fock space is the pre-Hilbert space $\displaystyle \Gamma = \oplus_{n=0}^\infty \mathbb{C} \Phi_n$. The operators $A, C$ are described as follows:
$A \Phi_0 = 0, A \Phi_n = \sqrt{n} \Phi_{n-1}, C \Phi_{n} = \sqrt{n + 1} \Phi_{n +1}$.
To consider the large quantum number limit, we consider the moment sequence 
\[ \left\{ \left\langle \left( \dfrac{X}{\sqrt{2k + 1}} \right)^m \Phi_k, \Phi_k \right\rangle \right\}_{m = 1}^\infty. \]
Here we divide $X$ by $\sqrt{2k + 1}$ to make the second moment $1$.
We have already shown in \cite[Theorem 3.1]{SAI} that as the natural number $k$ tends to infinity, each moments tend to those of the arcsine law
\[
\int_{-\sqrt{2}}^{\sqrt{2}} x^m \frac{dx}{\pi \sqrt{2-x^2}}.
\]
As a corollary, asymptotic behavior of the Hermite polynomials can be described by the arcsine law. 

This convergence to the arcsine law occurs not only in the Gaussian case but also in many other cases.
To find a conceptual reason, we focus on the commutation relation of the operators $A$, $B$, $C$. In the Gaussian case, the operators $A$ and $C$ satisfy some asymptotic commutation relation modulo variance. 
This weak form of commutativity is formulated in section \ref{Section:RAC}.
We prove in Theorem \ref{Theorem:RAC1} that the arcsine law appears as the large quantum number limits under the commutativity.
The assumption can be seen as one of the reasons why the arcsine law often appears as classical limit of the interacting Fock spaces and of orthogonal polynomials.

In the last part of this paper, we introduce a much weaker condition on the commutativity of the three processes $A$, $B$, $C$. This condition provides us discretized perturbations of the arcsine law.

\section{Basic notions}
\subsection{Algebraic Probability Space}
Let $\mathcal{A}$ be a $\ast$-algebra. We call a linear map $\varphi :\mathcal{A}\rightarrow \mathbb{C}$ a state on $\mathcal{A}$ if it satisfies
\[
\varphi(1)=1,\:\:\:\varphi (X^{\ast} X)\geq 0, \quad \mathrm{for}\ X \in \mathcal{A}.
\]
A pair $(\mathcal{A}, \varphi)$ of a $\ast$-algebra and a state on it is called an algebraic probability space. An element of $\mathcal{A}$ is called an algebraic random variable.
We define a notation for a state $\varphi :\mathcal{A}\rightarrow \mathbb{C}$, an element $X\in \mathcal{A}$ and a probability distribution $\mu$ on $\mathbb{R}$ as follows.
\begin{nota} We use the notation 
$X\sim_{\varphi} \mu $
when $\displaystyle \varphi(X^m) = \int_{\mathbb{R}}x^m\mu (dx)\:\:\: $for all $m \in \mathbb{N}$.
This stands for the identity between two moment sequences.
\end{nota}

\subsection{Interacting Fock space}

\begin{df}[Jacobi sequence] A pair of sequences $( \{\omega_{n + 1 /2}\}, \{\alpha_n\})$ is called a Jacobi sequence, 
\begin{itemize}
\item
if 
$\{\omega_{n + 1 /2}\}$ are positive real numbers $0 < \omega_{1/2}$, $\omega_{3/2}$, $\omega_{5/2}$, $\cdots$ labeled by half natural numbers, and
\item
if $\{\alpha_n\}$ are real numbers $\alpha_0$, $\alpha_1$, $\alpha_2$, $\cdots$ labeled by natural numbers. 
\end{itemize}
\end{df}

In other literatures as \cite[Definition 1.24]{H-O}, the sequence $\{\omega_{n + 1 /2}\}$ is called a Jacobi sequence of infinite type and given different labels.

\begin{df}[Interacting Fock space] Let $( \{\omega_{n + 1 /2}\}, \{\alpha_n\} )$ be a Jacobi sequence. An interacting Fock space $\Gamma_{\omega, \alpha}$ is a quadruple 
$(\Gamma(\mathbb{C}), A, B, C)$ where $\Gamma(\mathbb{C})$ is a pre-Hilbert space 
$\Gamma(\mathbb{C}):=\oplus^{\infty}_{n=0} \mathbb{C}\Phi _n$ with inner product given by $\langle\Phi_n, \Phi_m\rangle=\delta _{n,m}$, 
and $A, B, C$ are operators defined as follows:
\begin{itemize}
\item
$A$ is the annihilation operator
$A \Phi_0 =0$, $A \Phi_n=\sqrt{\omega_{n - 1/2}} \, \Phi_{n-1}.$
\item
$B$ is the preservation operator
$B \Phi_n=\alpha_{n} \Phi_n.$
\item
$C$ is the creation operator
$C \Phi_n =\sqrt{\omega_{n+1/2}} \, \Phi_{n + 1}.$
\end{itemize}
\end{df}

\begin{df}
The summation $X = A + B + C$ is expressed by the following symmetric tridiagonal matrix:
\begin{eqnarray*}
X =
\left(
\begin{array}{cccccccc}
\alpha_0                & \sqrt{\omega_{1/2}} & 0                           & \cdots \\
\sqrt{\omega_{1/2}} & \alpha_1                & \sqrt{\omega_{3/2}}  & \ddots  \\
0                          & \sqrt{\omega_{3/2}} & \alpha_2                 & \ddots \\
\vdots                  &  \ddots                  & \ddots                   & \ddots \\
\end{array}
\right).
\end{eqnarray*}
This is called the Jacobi matrix.
\end{df}
Accardi and Bo\.{z}ejko proved that every moment sequence 
$\displaystyle M_m = \int_{\mathbb{R}} x^m dx$ can be realized as that of an interacting Fock space $\langle X^m \Phi_0, \Phi_0 \rangle.$

Let $\mathcal{A}$ be the ${\ast}$-algebra generated by the matrices $A , B = B^*, C=(A)^{\ast}$ acting on the linear space $\oplus^{\infty}_{n=0} \mathbb{C}\Phi _n$.
Let $\varphi_k$ be the state defined as $\varphi_k(\cdot):=\langle\ \cdot \ \Phi_k, \Phi_k \rangle$. 
Then $(\mathcal{A}, \varphi_k)$ are algebraic 
probability spaces labeled by $k$. 
The asymptotic behavior of the sequence $\{(\mathcal{A}, \varphi_k)\}$ is the subject of this paper.

\subsection{Interacting Fock spaces and orthogonal polynomials}

Theorems for interacting Fock spaces often have interesting  interpretation in terms of orthogonal polynomials. To see this, we review the relation between interacting Fock spaces, probability measures and orthogonal polynomials. 
Let $\mu$ be a probability measure on $\mathbb{R}$ having finite moments. 
Then the space of polynomial functions is contained in the Hilbert space $L^2(\mathbb{R},\mu )$. 
The Gram-Schmidt procedure which provides orthogonal polynomials only depend on the moment sequence.

Let $\{p_n(x) \}_{n=0, 1, \cdots}$ be the monic orthogonal polynomials of $\mu$ such that 
the degree of $p_n$ equals to $n$. 
Then a relation among consecutive three terms
\begin{eqnarray*}
p_0 (x) &=& 1,\\
x p_0 (x) &=& p_1(x) + \alpha_0 p_0 (x),\\
x p_n(x) 
&=& p_{n+1}(x) + \alpha_{n}\, p_n(x) + \omega_{n - 1/2} \, p_{n-1}(x), \quad n \ge 1
\end{eqnarray*}
holds, if we appropriately choose the real numbers $\alpha_n$, $\omega_{n - 1/2}$.
It is not hard to prove that $\omega_{n - 1/2}$ are positive, if the support of $\mu$ is an infinite set.
Thus we obtain a Jacobi sequence
$(\{\omega_{n + 1/2}\}, \{\alpha_n\})$ out of the measure $\mu$.

Let $P_n$ denote the normalized orthogonal polynomial $p_n / \|p_n\|_2$.
It has been proved that the isometry 
$U:\Gamma_{\omega, \alpha} \to L^2(\mathbb{R},\mu ) : \Phi_n \mapsto P_n$ satisfies that $U^* x U = A + B + C$, where $x$ stands for the multiplication operator acting on $L^2(\mathbb{R},\mu )$. See \cite[Theorem 1.51]{H-O} for the proof. 
This means that we can decompose a measure-theoretic random variable into the sum of non-commutative algebraic random variables. 
This crucial idea in algebraic probability theory is called ``quantum decomposition'' in \cite{HAS} (see also \cite[Section 1.5]{H-O}).
Through the equality $U^* x U = A + B + C$, we obtain the identity in the moments
$A + B + C \sim _{\varphi_n} |P_n(x)|^2\mu(dx).$

\begin{rem}
For every algebraic probability space $(\cal{A}, \varphi)$ and every algebraic random variable $X \in \cal{A}$, 
it is known that there exists a probability measure $\mu$ on $\mathbb{R}$ which satisfies $X\sim_{\varphi} \mu$. 
\end{rem}

\section{Quantum-Classical Correspondence for Harmonic Oscillator}
\label{Section:QHO}

The interacting Fock space corresponding to $\omega_{n + 1/2} = n + 1, \alpha_n = 0$ is called ``Quantum Harmonic Osillator''. For quantum harmonic oscillator, it is well known that 
\[
X:= A + B + C = A + C
\]
represents the ``position'' and that 
\[
X\sim _{\varphi_0} \frac{1}{\sqrt{2\pi}} \exp\left(-\frac{x^{2}}{2} \right) dx.
\]
That is, in $n=0$ case the distribution of position is Gaussian. 

The asymptotic behavior of the distributions of position was nontrivial. What is the ``Classical limit'' of quantum harmonic oscillator?
This question, 
which is related to fundamental problems in quantum theory and asymptotic analysis \cite{EZA}, 
was analyzed in \cite[Section 3]{SAI} from the viewpoint of non-commutative algebraic probability 
with quite a simple combinatorial argument. The answer is nothing but the arcsine law.

\begin{thm}[Theorem 3.1 in \cite{SAI}]\label{ClimitQHO}
 Let $\Gamma_{\omega, \alpha}:=(\Gamma(\mathbb{C}), A, B\equiv 0, C)$ 
be the Quantum harmonic oscillator, $X:= A + C$ 
and $\mu_n$ be a probability distribution on $\mathbb{R}$ such that 
\[
\frac{X}{\sqrt{2k + 1}}\sim _{\varphi_k} \mu_k .
\]
Then $\mu_n$ weakly converges to the arcsine law $\mu_{As}$.
\end{thm}

Here $\sqrt{2k + 1}$ is the normalization factor to make the variance one, that is, $\left\langle \left( \dfrac{X}{\sqrt{2k + 1}} \right)^2 \Phi_k, \Phi_k \right\rangle = 1$. 
Since it is easy to see that the arcsine law gives ``time-averaged behavior'' of classical harmonic oscillator, the result can be viewed as ``Quantum-Classical Correspondence'' for harmonic oscillators.

As the case for the quantum harmonic osillator, we define the notion of classical limit distribution for interacting Fock spaces. 
It is a distribution to which the distribution for $X$ under $\varphi_n$, after normalization, 
converges in moment. 

\begin{df}[Classical Limit distribution]\label{Definition:CLD}
Let $\Gamma_{\omega, \alpha}:=(\Gamma(\mathbb{C}), A, B, C)$ be an interacting Fock space, $X$ be $A + B + C$ 
and $\mu_n$ be a probability distribution on $\mathbb{R}$ such that 
\[
\frac{X-\alpha_n}{\sqrt{\omega_{n + 1/2} + \omega_{n-1/2}}} \sim _{\varphi_n} \mu_n .
\]
A probability distribution $\mu$ on $\mathbb{R}$ is called a classical limit distribution of $\Gamma_{\omega, \alpha}$, if $\mu_n$ converge $\mu$ in moment.
\end{df}

By the normalizations $-\alpha_n$ and $\cdot / \sqrt{\omega_{n + 1/2} + \omega_{n-1/2}}$, the measure $\mu_n$ has mean $0$ and variance $1$.

\begin{rem}
Existence of a classical limit distribution depends on the Jacobi sequence $(\omega, \alpha)$.
In many cases which historically attract attention, the limit exists. See Remark \ref{Remark:ManyExamples}.

Uniqueness of classical limit distribution relates to the moment problem. 
Convergence in moment implies weak convergence in the case that the limit distribution is the 
solution of a determinate moment problem \cite{CHU,FEL}.
\end{rem}

A classical limit distribution on an interacting Fock space is also a weak limit of measures defined by square of orthogonal polynomials. For example, in the case of Gaussian distribution, Theorem \ref{ClimitQHO} implies the following.
Let $P_k$ be the sequence of normalized Hermite polynomials. Then $\displaystyle 
\left| P_k \left( x \right) \right|^2 
\frac{\exp(- x^{2} / 2)}{\sqrt{2\pi}} dx$ defines a sequence of probability measures whose second moment is $2k + 1$.
The sequence of normalizations
$$\displaystyle 
\sqrt{2k + 1} \left| P_k \left( \sqrt{2k + 1} x \right) \right|^2 
\frac{\exp(- (2k + 1)x^{2} / 2)}{\sqrt{2\pi}} dx$$ 
weakly converges to the arcsine law $\displaystyle \mu_{As} (dx)= \frac{dx}{\pi \sqrt{2-x^2}}$.

\section{Two-sided infinite Jacobi sequences}
In this section, we set up the framework to analyze classical limit distributions.
We introduce two-sided infinite Jacobi sequences.

\begin{df}[Two-sided Jacobi sequence] Let 
\begin{eqnarray*}
\omega &=& 
\left\{ \omega_m \ge 0 \ \left| \  m = \cdots, -\dfrac{3}{2}, -\dfrac{1}{2}, \dfrac{1}{2}, \dfrac{3}{2}, \cdots \right.\right\}, \\
\alpha &=& \{\alpha_n \in \mathbb{R} \ |\ n = \cdots, -2, -1, 0, 1, \cdots\}.
\end{eqnarray*}
be two-sided infinite sequences of reals satisfying one of the following conditions $(1)$ or $(2)$:
\begin{enumerate}
\item
There exists a non-positive integer $N$ such that 
\begin{itemize}
\item
if $m < N$, then $\omega_m = 0$,
\item
if $m > N$, then $\omega_m > 0$,
\item
and if $n < N$, then $\alpha_n = 0$.
\end{itemize}
\item
For every half integers $m = \cdots, -\dfrac{3}{2}, -\dfrac{1}{2}, \dfrac{1}{2}, \dfrac{3}{2}, \cdots$, we have $\omega_m >0$.
\end{enumerate}
We call the pair $(\omega, \alpha)$ a two-sided Jacobi sequence.
\end{df}

\begin{df}[Two-sided interacting Fock space] Let $( \omega, \alpha )$ be a two-sided Jacobi sequence. An interacting Fock space $\Gamma_{\omega, \alpha}$ is a quadruple 
$(\Gamma(\mathbb{C}), A, B, C)$ consists of
a pre-Hilbert space $\Gamma(\mathbb{C}) = \displaystyle\oplus^{\infty}_{n=-\infty} \mathbb{C}\Phi _n$ with inner product given by $\langle\Phi_n, \Phi_m\rangle=\delta _{n,m}$, 
and operators $A, B, C$ defined as follows:
\begin{itemize}
\item
$A$ is the annihilation operator
$A \Phi_n=\sqrt{\omega_{n - 1/2}} \, \Phi_{n-1}.$
\item
$B$ is the preservation operator
$B \Phi_n=\alpha_{n} \Phi_n.$
\item
$C$ is the creation operator
$C \Phi _n=\sqrt{\omega_{n+1/2}} \, \Phi_{n+1}.$
\end{itemize}
\end{df}

\begin{df}
The summation $X = A + B + C$ is expressed by the tridiagonal matrix
$X = [X_{m, n}]_{m,n \in \mathbb{Z}}$ whose matrix coefficients are given by:
\begin{eqnarray*}
X_{m,n} = 
\left\{
\begin{array}{rcl}
\sqrt{\omega_{n - 1/2}}, & & m = n - 1,\\
\alpha_n, & & m = n,\\
\sqrt{\omega_{n + 1/2}}, & & m = n + 1,\\
0, & & |m - n| \ge 2.
\end{array}
\right.
\end{eqnarray*}
This operator $X$ is called the two-sided Jacobi matrix of $(\omega, \alpha)$.
\end{df}

The matrix $X$ is an algebraic random variable. Its moments with respect to the state $\langle \cdot \Phi_0, \Phi_0\rangle$ can be described by the matrix entries as follows:
\begin{eqnarray*}
\left\langle X^1 \Phi_0, \Phi_0 \right\rangle &=& \alpha_0,\\
\left\langle X^2 \Phi_0, \Phi_0 \right\rangle &=& \omega_{-1/2} + \alpha_0^2 + \omega_{-1/2},\\
\left\langle X^3 \Phi_0, \Phi_0 \right\rangle &=& \omega_{-1/2} \alpha_{-1} + 2 \omega_{-1/2} \alpha_0 + \alpha_0^3 + 2 \omega_{1/2} \alpha_0 + \omega_{1/2} \alpha_1,\\
&\cdots&.
\end{eqnarray*}

It is easily shown by induction that the matrix coefficients of $X^m$ are described by polynomials 
of $\{\omega_{n + 1/2}\} \cup \{\alpha_n\}$. Therefore we have the following lemma.

\begin{lem}\label{LemmaMomentLimit}
Let $\left\{ \left( \omega^{(k)}, \alpha^{(k)} \right) \right\}_k$ be a sequence of two-sided Jacobi sequences and 
let $(\omega, \alpha)$ be a two-sided Jacobi sequence. 
Let $X^{(k)}$ and $X$ be the corresponding Jacobi matrices acting on $\oplus_{n = -\infty}^\infty \mathbb{C} \Phi_n$. 
If $\lim_{k \to \infty} \omega_{n + 1/2}^{(k)} = \omega_{n + 1/2}$ and $\lim_{k \to \infty} \alpha_n^{(k)} = \alpha_n$ for every integer $n$, 
then we have the following moment convergence: 
$\displaystyle \lim_{k \to \infty} \left\langle \left( X^{(k)} \right)^m \Phi_0, \Phi_0 \right\rangle = \langle X^m \Phi_0, \Phi_0 \rangle.$
\end{lem}
%

\section{The Arcsine law as classical limit distribution}\label{Section:RAC}

\subsection{Relative asymptotic commutativity (RAC1) }

In this part, we propose a condition (RAC1) for the one-sided interacting Fock space $\Gamma_{\omega, \alpha}$. The conditions handle asymptotic behavior of the creation $C$, the preservation $B$, and the annihilation $A$  modulo the standard variance $\omega_{n + 1/2} + \omega_{n - 1/2}$.

\begin{df}
The interacting Fock space is said to satisfy {\rm (RAC1)}, 
if the commutators $[A, C]$ and $[A, B]$ are asymptotically zero in the following sense:
\[ \lim_{n \to \infty } \dfrac{ (A C - C A) \Phi_n }{\omega_{n + 1/2} + \omega_{n - 1/2}} = 0,\quad 
\lim_{n \to \infty } \dfrac{ (A B - B A) \Phi_n}{\omega_{n + 1/2} + \omega_{n - 1/2}} = 0.\]
\end{df}
Recall that $\langle \ \cdot \ \Phi_n, \Phi_n \rangle$ stands for the $n$-th state of the interacting Fock space.

\begin{lem}\label{LemmaRAC1}
The condition {\rm (RAC1)} is equivalent to
\[\lim_{n \to \infty } \frac{\omega_{n + 1/2}}{\omega_{n-1/2}} = 1, \quad
\lim_{n \to \infty } \frac{\alpha_n - \alpha_{n - 1}}
{\sqrt{\omega_{n + 1/2} + \omega_{n - 1/2}}} = 0.\]
\end{lem}

\begin{proof}
The commutators $[A, C]$ and $[A, B]$ satisfy the following:
\begin{eqnarray*}
\dfrac{A C - C A}{\omega_{n + 1/2} + \omega_{n - 1/2}}\, \Phi_n 
&=& \frac{\omega_{n + 1/2} - \omega_{n - 1/2}}{\omega_{n + 1/2} + \omega_{n - 1/2}} \Phi_n,\\
\dfrac{A B - B A}{\omega_{n + 1/2} + \omega_{n - 1/2}} \Phi_k 
&=& \frac{\alpha_n - \alpha_{n - 1}}{\omega_{n + 1/2} + \omega_{n - 1/2}} \sqrt{\omega_{n-1/2}} \Phi_{n-1}.
\end{eqnarray*}
It is not difficult to show that if the conditions in the lemma holds, then the above vectors converge to $0$.

Conversely, we suppose that (RAC1) holds.
In this case, we have
\begin{eqnarray*}
\lim_{n \to \infty} \frac{\omega_{n + 1/2} - \omega_{n - 1/2}}{\omega_{n + 1/2} + \omega_{n - 1/2}} = 0, \quad
\lim_{n \to \infty} \frac{\alpha_n - \alpha_{n - 1}}{\omega_{n + 1/2} + \omega_{n - 1/2}} \sqrt{\omega_{n-1/2}} = 0.
\end{eqnarray*}
By the equality
\[
\frac{2}{1 - \dfrac{\omega_{n + 1/2} - \omega_{n - 1/2}}{\omega_{n + 1/2} + \omega_{n - 1/2}}} -1 =
\frac{\omega_{n + 1/2}}{\omega_{n - 1/2}},
\]
the first condition of (RAC1) implies that
\begin{eqnarray*}
\lim_{n \to \infty} \frac{\omega_{n + 1/2}}{\omega_{n - 1/2}} =\frac{2}{1 - 0} - 1 = 1.
\end{eqnarray*}
The second condition means that
\begin{eqnarray*}
\lim_{n \to \infty }\frac{\alpha_n - \alpha_{n - 1}}{\omega_{n + 1/2} + \omega_{n - 1/2}} \sqrt{\omega_{n-1/2}} = 0.
\end{eqnarray*}
It follows that
\begin{eqnarray*}
&&\lim_{n \to \infty } \frac{\alpha_n - \alpha_{n - 1}}
{\sqrt{\omega_{n + 1/2} + \omega_{n - 1/2}}} \\
&=&
\lim_{n \to \infty } \frac{\alpha_n - \alpha_{n - 1}}
{\omega_{n + 1/2} + \omega_{n - 1/2}} \sqrt{\omega_{n-1/2}}
\cdot
\lim_{n \to \infty } \sqrt{\frac{\omega_{n + 1/2} + \omega_{n - 1/2}}{\omega_{n-1/2}}}\\
&=& 0 \ \sqrt{2} = 0.
\end{eqnarray*}
Now we obtain the conditions in the lemma.
\end{proof}

The quantum harmonic oscillator introduced in section \ref{Section:QHO} satisfies the above condition. 
The following theorem is the main result in this paper and a generalization of Theorem \ref{ClimitQHO}.
\begin{thm}\label{Theorem:RAC1}
Let $\Gamma_{\omega, \alpha}:=(\Gamma(\mathbb{C}), A, B, C)$
be an interacting Fock space satisfying asymptotic commutativity {\rm (RAC1)}.
Then the classical limit distribution given in Definition $\ref{Definition:CLD}$ exists and is the arcsine law $\dfrac{dx}{\pi \sqrt{2 - x^2}}$. 
\end{thm}

\begin{proof}
Let $(\{\omega_{n + 1/2}\}, \{\alpha_n\})$ be a one-sided Jacobi sequence. Suppose that (RAC1) holds. Consider the $k$-th state $\langle \ \cdot\ \Phi_k, \Phi_k \rangle$ and the normalized algebraic random variable
\[ X^{(k)} = \dfrac{X - \alpha_k}{\sqrt{\omega_{k + 1/2} + \omega_{k - 1/2}}}
\]
acting on $\oplus_{n = 0}^\infty \mathbb{C} \Phi_n$.
The matrix coefficients are described by
\begin{eqnarray*}
X^{(k)}_{m, n} =
\left\{
\begin{array}{cl}
\dfrac{\omega_{n - 1/2}}{\sqrt{\omega_{k + 1/2} + \omega_{k - 1/2}}}, & m = n-1,\\
\dfrac{\alpha_n - \alpha_k}{\sqrt{\omega_{k + 1/2} + \omega_{k - 1/2}}}, & m = n,\\
\dfrac{\omega_{n + 1/2}}{\sqrt{\omega_{k + 1/2} + \omega_{k - 1/2}}}, & m = n + 1,\\
0, & |m - n| \ge 2.
\end{array}
\right.
\end{eqnarray*}
To study asymptotic behavior of $X^{(k)}$ acting on $\oplus_{n = 0}^\infty \mathbb{C} \Phi_n$, 
we change the index $m, n = 0, 1, \cdots, k, \cdots$ to $m, n = -k, -k + 1, \cdots, 0, \cdots$, and exploit two-sided interacting Fock space $\Gamma^{(k)} = \oplus_{n = -k}^\infty \mathbb{C} \Phi_n$.
We now consider the state $\langle \ \cdot\ \Phi_0, \Phi_0 \rangle$ and algebraic random variable $\widetilde{X^{(k)}}$ defined by
\begin{eqnarray*}
\widetilde{X^{(k)}}_{m, n} =
\left\{
\begin{array}{cl}
\dfrac{\omega_{n + k- 1/2}}{\sqrt{\omega_{k + 1/2} + \omega_{k - 1/2}}}, & m = n-1,\\
\dfrac{\alpha_{n + k} - \alpha_k}{\sqrt{\omega_{k + 1/2} + \omega_{k - 1/2}}}, & m = n,\\
\dfrac{\omega_{n + k + 1/2}}{\sqrt{\omega_{k + 1/2} + \omega_{k - 1/2}}}, & m = n + 1,\\
0, & |m - n| \ge 2.
\end{array}
\right.
\end{eqnarray*}
By the first condition of Lemma \ref{LemmaRAC1}, neighboring ratio of $\{\omega_{n + k + 1/2}\}_{n = -k}^\infty$ is $1$. This implies that for every fixed integer $n$, 
\[\lim_{k \to \infty} \widetilde{X^{(k)}}_{n - 1, n} = \dfrac{1}{\sqrt{2}} = \lim_{k \to \infty} \widetilde{X^{(k)}}_{n + 1, n}.\]
By the second condition of Lemma \ref{LemmaRAC1}, $\displaystyle \lim_{k \to \infty } \frac{\alpha_k - \alpha_{k - 1}}
{\sqrt{\omega_{k + 1/2} + \omega_{k - 1/2}}} = 0$. Together with $\displaystyle \lim_{k \to \infty} \dfrac{\omega_{n + k + 1/2}}{\omega_{k + 1/2}} = 0$, this implies that for every $n$, $\displaystyle \lim_{k \to \infty} \widetilde{X^{(k)}}_{n, n} = 0$.

Now we exploit Lemma \ref{LemmaMomentLimit}.
Let $\widetilde{X}$ be the two-sided infinite matrix
\begin{eqnarray*}
\left(
\begin{array}{cccccccccccc}
\ddots & \ddots &       &    \\
\ddots & 0         & 1/\sqrt{2} &   &  \\
          & 1/\sqrt{2} & {\bf 0} & 1/\sqrt{2} & \\
          &                & 1/\sqrt{2} & 0 & \ddots & \\
           &                &                & \ddots & \ddots  \\
\end{array}
\right)
\end{eqnarray*}
acting on $\ell_2(\mathbb{Z})$. The bold zero $\bf 0$ stands for the position of the matrix coefficient $\langle \ \cdot \ \Phi_0, \Phi_0 \rangle$. By Lemma \ref{LemmaMomentLimit}, we have
\[
\lim_{k \to \infty} \left\langle \left( X^{(k)} \right)^m \Phi_k, \Phi_k \right\rangle
= 
\lim_{k \to \infty} \left\langle \left( \widetilde{X^{(k)}} \right)^m \Phi_0, \Phi_0 \right\rangle
= 
\left\langle \left( \widetilde{X} \right)^m \Phi_0, \Phi_0 \right\rangle.
\]
Via the Fourier transform $\ell^2(\mathbb{Z}) \cong L^2(\{e^{it}\})$,
we can identify the vector $\Phi_0$ with the constant function $1$ on the circle $\mathbb{T} = \{e^{it}\}$, and the operator $\widetilde{X}$ with the multiplication operator
\[\dfrac{1}{\sqrt{2}} e^{it} + \dfrac{1}{\sqrt{2}} e^{-it} = \sqrt{2} \cos t.\] 
Thus we have
\begin{eqnarray*}
\left\langle \left( \widetilde{X} \right)^m \Phi_0, \Phi_0 \right\rangle_{\ell^2(\mathbb{Z})}
&=&
\left\langle \left( \sqrt{2} \cos t \right)^m 1, 1 \right\rangle_{L^2(\{e^{it}\})} \\
&=&
\int_{0}^{2\pi} \left( \sqrt{2} \cos t \right)^m \frac{dt}{2 \pi}\\
&=&
\int_{\pi}^{2\pi} \left( \sqrt{2} \cos t \right)^m \frac{dt}{\pi}.
\end{eqnarray*}
Replacing $\sqrt{2} \cos t$ with $x$, we have
\[
\lim_{k \to \infty} \left\langle \left( X^{(k)} \right)^m \Phi_k, \Phi_k \right\rangle
= 
\int_{- \sqrt{2}}^{\sqrt{2}} x^m \frac{dx}{\pi \sqrt{2 - x^2}}.
\]
\end{proof}

\begin{rem}\label{Remark:ManyExamples}
The theorem means that the arcsine law is turned out to be the classical limit distribution in many cases. We pick up several examples.
\begin{enumerate}
\item
The interacting Fock spaces corresponding to uniform distribution $\chi_{[-1, 1]} dx$,
is described by the Jacobi sequence
\[\omega_{n + 1/2} = \frac{(n + 1)^2}{(2n + 1)(2n + 3)}, \quad \alpha_n = 0.\] 
\item
The quantum decomposition of the exponential distribution $\chi_{[0, \infty)} e^{-x} dx$  
is given by the Jacobi sequence
\[\omega_{n + 1/2} = (n + 1)^2, \quad \alpha_n = 2n + 1.\] 
\item
$q$-Gaussians ($-1< q \leq 1$) are probability measure on $\mathbb{R}$ given by the Jacobi sequence
\[\omega_{n + 1/2} =1+q+q^2+\cdots +q^{n}, \quad \alpha_n = 0.\] 
The case of $q=1$ corresponds to the Gaussian measure. The case of $q=0$ corresponds to the semicircle law $\dfrac{\sqrt{4 - x^2}\ dx}{2 \pi}$ of Wigner.
\end{enumerate}
By Lemma \ref{LemmaRAC1}, these interacting Fock spaces satisfy (RAC1).
\end{rem}

\begin{rem}
Since the arcsine law is the solution of a determinate 
moment problem, moment convergence implies weak convergence.
\end{rem}

\begin{rem}
It is quite interesting to compare Kerov's theorem on his ``Arcsine Law''which is different from the probability measure $\dfrac{dx}{\pi \sqrt{2 - x^2}}$ but closely related to it \cite{KER}.
\end{rem}

Theorem \ref{Theorem:RAC1} implies the following asymptotic behavior of orthogonal polinomials:
\begin{cor}
Let  $\mu$ be a prabablity measure such that the corresponting 
Jacobi sequence $(\{\omega_n\}, \{\alpha_n\})$ satisfies the conditions above.
Let $P_n$ be the normalized orthogonal polynomial with degree $n$.
The measure $\mu_n$ defined as 
$\mu_n(dx):=|P_n(\sqrt{\omega_{n + 1/2}+\omega_{n - 1/2}}x)|^2\mu(\sqrt{\omega_{n + 1/2}+\omega_{n - 1/2}}dx)$ weakly converge to the arcsine law $\mu_{As}$.
\end{cor}

Many kinds of orthogonal polynomials such as Legendre polynomials, Laguerre polynomials or $q$-Hermite polynomials for $-1< q\leq 1$ satisfy the above condition.

\section{Weaker form of asymptotic commutativity and Classical limits}

It is reasonable to guess that we can obtain other types of classical limits assuming weaker condition on the operators $A, B, C$. Relaxing the commutativity condition between $A$ and $B$, we have discretized arcsine laws as classical limits.

\begin{df}
The interacting Fock space is said to satisfy {\rm (RAC2)}, 
if the commutator $[A, C]$ is asymptotically zero and if $[A, B]$ is asymptotically a scalar multiple of $A$ in the following sense:
\begin{itemize}
\item
$\displaystyle \lim_{n \to \infty } \dfrac{(A C - C A) \Phi_n}{\omega_{n + 1/2} + \omega_{n - 1/2}} = 0$
and
\item
there exists a real number $r$ satisfying
\[ \lim_{n \to \infty } \dfrac{[(A B - B A) - r A] \Phi_n}{\omega_{n + 1/2} + \omega_{n - 1/2}} = 0.\]
\end{itemize}
\end{df}
Recall that $\omega_{n + 1/2} + \omega_{n - 1/2}$ is the variance of $X = A + B + C$ with respect to $\langle \ \cdot \ \Phi_n, \Phi_n \rangle$.
The proof of the following is not so hard.

\begin{lem}\label{Lemma:RAC2}
The condition {\rm (RAC2)} is equivalent to
$\displaystyle \lim_{n \to \infty } \frac{\omega_{n+1/2}}{\omega_{n- 1/2}} = 1$ and
convergence of the sequence $\displaystyle \left\{ \frac{\alpha_{n}-\alpha_{n-1}}{\sqrt{\omega_{n + 1/2} + \omega_{n - 1/2}}} \right\}_n$.
\end{lem}

In the following subsection, we denote by $c$ the limit of the latter sequence.

\begin{exa}
\begin{itemize}
\item
An interacting Fock space with (RAC1) satisfies (RAC2).
\item
The one-sided interacting Fock space $\Gamma_{\omega, \alpha}$ defined by $\omega_{n + 1/2} = 1/2$ and $\alpha_n = c n$ shares the property (RAC2). 
The infinite Jacobi matrix is given
\begin{eqnarray*}
X =
\left(
\begin{array}{ccccccccc}
0             & 1/\sqrt{2} & 0              & \\
1/\sqrt{2} & c             & 1/\sqrt{2}  & \ddots            \\
0             & 1/\sqrt{2} & 2c             & \ddots \\
               &\ddots      & \ddots      & \ddots \\
\end{array}
\right)
\end{eqnarray*}
\end{itemize}
\end{exa}

\subsection{Calculation of the classical limits}

From now on, we consider the case that the interacting Fock space satisfy (RAC2).
Let $X^{(k)}$ be the random variable $\dfrac{X - \alpha_k}{\sqrt{\omega_{k + 1/2} + \omega_{k - 1 / 2}}}$.
Observing the Jacobi sequence, we obtain the following lemma.

\begin{lem}\label{Lemma:RAC2Coefficient}
For every integers $m, n$, we have
\begin{eqnarray*}
\lim_{k \to \infty} \left\langle X^{(k)} \Phi_{k + m}, \Phi_{k + n} \right\rangle =
\left\{
\begin{array}{cl}
1/ \sqrt{2}, & m = n - 1,\\
cn, & m = n,\\
1/ \sqrt{2}, & m = n + 1,\\
0, & |m - n| \ge 2.\\
\end{array}
\right.
\end{eqnarray*}
\end{lem}

\begin{proof}
The proof is similar to the first half of the proof of Theorem \ref{Theorem:RAC1}. By the definition of the operators $A, B, C$, the operator $X^{(k)}$ satisfies 
\begin{eqnarray*}
\left\langle X^{(k)} \Phi_{k + m}, \Phi_{k + n} \right\rangle =
\left\{
\begin{array}{cl}
\dfrac{\omega_{n + k- 1/2}}{\sqrt{\omega_{k + 1/2} + \omega_{k - 1/2}}}, & m = n-1,\\
\dfrac{\alpha_{n + k} - \alpha_k}{\sqrt{\omega_{k + 1/2} + \omega_{k - 1/2}}}, & m = n,\\
\dfrac{\omega_{n + k + 1/2}}{\sqrt{\omega_{k + 1/2} + \omega_{k - 1/2}}}, & m = n + 1,\\
0, & |m - n| \ge 2.
\end{array}
\right.
\end{eqnarray*}
The condition in Lemma \ref{Lemma:RAC2} implies the above lemma.
\end{proof}

We grasp the asymptotic behavior of $X^{(k)}$ with respect to the state $\langle \ \cdot \ \Phi_k, \Phi_k \rangle$, using two-sided infinite tridiagonal matrices acting on the inner product space $\oplus_{k \in \mathbb{Z}} \mathbb{C} \Phi_k$. Putting the limit of the matrix coefficient $\langle X^{(k)} \Phi_{k + m}, \Phi_{k + n} \rangle$ at $(m, n)$-entry, we obtain the following tridiagonal operator:
\begin{eqnarray*}
\widetilde{X} =
\left(
\begin{array}{ccccccccc}
\ddots     &\ddots     &\ddots     & \\
\ddots     &-2c             & 1/\sqrt{2} & 0              & \\
\ddots     & 1/\sqrt{2} &-c             & 1/\sqrt{2} & 0              & \\
              &0             &1/\sqrt{2} &\bf 0             & 1/\sqrt{2} & 0              & \\
              &              &0             &1/\sqrt{2} & c             & 1/\sqrt{2}  & \ddots            \\
              &              &              &0             & 1/\sqrt{2} & 2c             & \ddots \\
              &              &              &               &\ddots      & \ddots      & \ddots \\
\end{array}
\right)
\end{eqnarray*}
``{\bf 0}'' is at the position of $(0,0)$. 
By Lemma \ref{LemmaMomentLimit}, convergence of the matrix coefficients implies the following moment convergence:
\begin{eqnarray}\label{equation:RAC2}
\lim_{n \to \infty} 
\left \langle \left(X^{(k)} \right)^m \, \Phi_k, \Phi_k \right \rangle =
\left \langle \widetilde{X}^m \, \Phi_{0}, \Phi_{0} \right\rangle.
\end{eqnarray}

To see the limit of $X^{(k)}$, we study the densely defined operator $\widetilde X$ acting on $\ell^2(\mathbb{Z}) = \oplus_{k \in \mathbb{Z}} \mathbb{C} \Phi_k$.
Via the Fourier transform $\ell^2(\mathbb{Z}) \cong L^2(\mathbb{T})$,
we may regard $\widetilde{X}$ as a densely defined symmetric operator acting on $L^2 \left( \{e^{it} \} \right)$. The space of Laurent polynomials of $z = e^{it}$ is the domain of $\widetilde{X}$. 
The operator $\widetilde{X}$ acts on the Laurent polynomials as follows:
\begin{itemize}
\item
the annihilation part of $\widetilde{X}$ is identified with the multiplication operator
$e^{-it}/ \sqrt{2}$,
\item
the diagonal part of $\widetilde{X}$ is identified with the differential operator
$\dfrac{c}{i} \dfrac{d}{dt}$.
\item
the creation part of $\widetilde{X}$ is identified with the multiplication operator
$e^{it}/ \sqrt{2}$,
\end{itemize}
In the case that $c \neq 0$,
the summation is expressed by
\begin{eqnarray*}
\frac{e^{-it}}{\sqrt{2}} + \dfrac{c}{i} \dfrac{d}{dt} + \frac{e^{it}}{\sqrt{2}} 
&=& \dfrac{c}{i} \left(\dfrac{d}{dt} + i \dfrac{\sqrt{2} \cos t}{c} \right).
\end{eqnarray*}
We may further calculate
\begin{eqnarray*}
\frac{e^{-it}}{\sqrt{2}} + \dfrac{c}{i} \dfrac{d}{dt} + \frac{e^{it}}{\sqrt{2}} 
&=& \dfrac{c}{i} \exp\left(- i \dfrac{\sqrt{2} \sin t}{c}\right) 
\circ \dfrac{d}{dt} \circ \exp\left(i \dfrac{\sqrt{2} \sin t}{c}\right)\\
&=& \exp\left(- i \dfrac{\sqrt{2} \sin t}{c}\right) 
\circ \left(\dfrac{c}{i} \dfrac{d}{dt}\right) \circ \exp\left(i \dfrac{\sqrt{2} \sin t}{c}\right).
\end{eqnarray*}
We can easily prove the above equation by hitting an arbitrary Laurent polynomial.
We note that the absolute value of $\exp\left(i \dfrac{\sqrt{2} \sin t}{c}\right)$ is $1$.
Define $a_n(c)$ by the Fourier expansion
\[\exp\left(i \dfrac{\sqrt{2} \sin t}{c}\right) 
= \sum_{n \in \mathbb{Z}} a_n(c) e^{i n t}.\]

\begin{df}
For $x \in \mathbb{R}$, we denote by $\delta_x$ the probability measure concentrated on $x$.
The probability measure 
\[\mu_c = \sum_{n \in \mathbb{Z}} |a_n(c)|^2 \delta_{cn} \]
on $\mathbb{R}$
is called a discrete arcsine distribution.
\end{df}

\begin{thm}

Suppose that the interacting Fock space $\Gamma_{\{\omega_n\}, \{\alpha_n\}}$ satisfy the condition {\rm (RAC2)} but does not satisfy {\rm (RAC1)}. Define a real number $c$ by 
$\displaystyle \lim_{n \to \infty} \frac{\alpha_{n}-\alpha_{n-1}}{\sqrt{\omega_{n - 1/2} + \omega_{n + 1/2}}} $. 
Then for each natural number $m$, we have the following moment convergence:
\[\lim_{k \to \infty} \left\langle \left( \dfrac{X - \alpha_k}{\sqrt{\omega_{k + 1/2} + \omega_{k - 1/2}}} \right)^m \, \Phi_k, \Phi_k \right \rangle = 
\int_{\mathbb{R}} x^m d \mu_c.\]
\end{thm}

\begin{proof}
By the equation (\ref{equation:RAC2}),
it suffices to show that
\[\left \langle \widetilde{X}^k \, \Phi_{0}, \Phi_{0} \right\rangle = 
\int_{\mathbb{R}} x^m d \mu_c.\]
Considering the Fourier transform, the left hand side is equal to
\begin{eqnarray*}
&&
\left\langle \left(\frac{e^{-it}}{\sqrt{2}} + \frac{e^{it}}{\sqrt{2}} + \dfrac{c}{i} \dfrac{d}{dt} \right)^m 1, 1 \right\rangle_{L^2(\{e^{it}\})} \\
&=& \left\langle \left\{ 
\exp\left(- i \dfrac{\sqrt{2} \sin t}{c}\right) 
\circ \left(\dfrac{c}{i} \dfrac{d}{dt}\right) \circ \exp\left(i \dfrac{\sqrt{2} \sin t}{c}\right)\right\}^m 1, 1\right\rangle_{L^2(\{e^{it}\})}\\
&=& \left\langle   
\left(\dfrac{c}{i} \dfrac{d}{dt}\right)^m \exp\left(i \dfrac{\sqrt{2} \sin t}{c}\right), \exp\left( i \dfrac{\sqrt{2} \sin t}{c}\right)\right\rangle_{L^2(\{e^{it}\})}.
\end{eqnarray*}
By the Fourier expansion of $\exp\left(i \dfrac{\sqrt{2} \sin t}{c}\right)$, the above quantity is
\begin{eqnarray*}
&& \left\langle   
\left(\dfrac{c}{i}\right)^m \dfrac{d^m}{dt^m} 
\exp\left(i \dfrac{\sqrt{2} \sin t}{c}\right), 
\sum_{n \in \mathbb{Z}} a_n(c) e^{int} \right\rangle_{L^2(\{e^{it}\})}\\
&=& 
\sum_{n \in \mathbb{Z}} 
\overline{a_n(c)}
\left\langle \left(\dfrac{c}{i}\right)^m \dfrac{d^m}{dt^m}
\exp\left(i \dfrac{\sqrt{2} \sin t}{c}\right), 
e^{int} \right\rangle_{L^2(\{e^{it}\})}.
\end{eqnarray*}
By iteration of partial integration, this is equal to
\begin{eqnarray*}
&& \sum_{n \in \mathbb{Z}} \overline{a_n(c)}
\left\langle 
\exp\left(i \dfrac{\sqrt{2} \sin t}{c}\right), 
\left(\dfrac{c}{i}\right)^m \dfrac{d^m}{dt^m}e^{int} \right\rangle_{L^2(\{e^{it}\})}\\
&=& \sum_{n \in \mathbb{Z}} (cn)^m \overline{a_n(c)}
\left\langle 
\exp\left(i \dfrac{\sqrt{2} \sin t}{c}\right), 
e^{i n t} \right\rangle_{L^2(\{e^{it}\})}\\
&=& \sum_{n \in \mathbb{Z}} (cn)^m |a_n(c)|^2.
\end{eqnarray*}
This is nothing other than $\displaystyle \int_{\mathbb{R}} x^m d \mu_c$.
\end{proof}

\subsection{Calculation of the discrete arcsine law $\mu_c$}

To identify the discrete arcsine law $\mu_c$, we have only to calculate the Fourier expansion of $\exp\left(i \dfrac{\sqrt{2} \sin t}{c}\right)$. By the Maclaurin expansion of the exponential function, we have
\begin{eqnarray*}
\exp\left(i \dfrac{\sqrt{2} \sin t}{c}\right)
&=& \exp \left(\dfrac{e^{it} - e^{-it}}{\sqrt{2} c} \right)\\
&=& \sum_{k = 0}^\infty \frac{1}{k !} \left(\dfrac{e^{it} - e^{-it}}{\sqrt{2} c} \right)^k.
\end{eqnarray*}
By the binomial theorem, we have
\begin{eqnarray*}
\exp\left(i \dfrac{\sqrt{2} \sin t}{c}\right)
&=& \sum_{k = 0}^\infty \frac{1}{k !} 
\sum_{l = 0}^k 
\frac{(-1)^l}{(\sqrt{2} c)^k}
\left(
\begin{array}{c}
k\\
l
\end{array}
\right) e^{i(k-l)t} e^{-i l t}\\
&=& \sum_{k = 0}^\infty
\sum_{l = 0}^k 
\frac{(-1)^l}{(\sqrt{2} c)^k}
 \frac{1}{l ! (k - l) !}  e^{i(k-2l)t}.
\end{eqnarray*}
It is not hard to check that this summation of the absolute values uniformly converges. Therefore it is possible to change the order of summation.
Now we define $n$ by $k - 2l$. 
The condition $0 \le l \le k$ is described by $0 \le l \le n + 2l$. This is equivalent to $0 \le l, - n \le l$. 
Then the Fourier expansion is described by
\begin{eqnarray*}
\exp\left(i \dfrac{\sqrt{2} \sin t}{c}\right)
&=& \sum_{n = -\infty}^\infty
\sum_{l = \max\{0, - n \}}^\infty
\frac{(-1)^l}{(\sqrt{2} c)^{n + 2l}}
 \frac{1}{l ! (n + l) !}  e^{i n t}.
\end{eqnarray*}
For $n \ge 0$, we have
\begin{eqnarray*}
a_n(c) &=&
\sum_{l = 0}^\infty
\frac{(-1)^l}{(\sqrt{2} c)^{n + 2l}} \frac{1}{l ! (n + l) !},\\
a_{-n}(c) &=&
\sum_{l = n}^\infty
\frac{(-1)^l}{(\sqrt{2} c)^{-n + 2l}} \frac{1}{l ! (-n + l) !}\\
&=&\sum_{l = 0}^\infty
\frac{(-1)^{l + n}}{(\sqrt{2} c)^{n + 2l}} \frac{1}{(n + l) ! l !}.\\
\end{eqnarray*}

\begin{thm}
The discrete arcsine law $\mu_c$ is a probability measure supported on $c \mathbb{Z}$. For $n = 0, 1, 2, \cdots$, the weight at $cn$ and $-cn$ is
\begin{eqnarray*}
\mu_c(\{cn\}) = \mu_c(\{-cn\}) = \frac{1}{2^n c^{2n}}
\left(\sum_{l = 0}^\infty \frac{(-1)^{l}}{(\sqrt{2} c)^{2l}} \frac{1}{(n + l) ! l !}\right)^2.
\end{eqnarray*}
 
\end{thm}

\subsection{Remarks on the discrete arcsine law}

Before closing this subsection, let us consider the limit of $\mu_c$ as $c \to 0$.
The $m$-th moment of the discrete arcsine $\mu_c$ is given by
\[\left\langle \left(\frac{e^{-it}}{\sqrt{2}} + \frac{e^{it}}{\sqrt{2}} + \dfrac{c}{i} \dfrac{d}{dt} \right)^m 1, 1 \right\rangle_{L^2(\{e^{it}\})}.\]
When $c$ goes to $0$, the moment converges to
\[\left\langle \left(\frac{e^{-it}}{\sqrt{2}} + \frac{e^{it}}{\sqrt{2}} \right)^m 1, 1 \right\rangle_{L^2(\{e^{it}\})} = \int_{-\pi}^\pi \left(\sqrt{2} \cos t \right)^m \frac{dt}{2\pi}
=
\int_{-\sqrt 2}^{\sqrt 2} x^m \dfrac{dx}{\pi \sqrt{2 - x^2}}.
\]
This is the $k$-th moment of the arcsine law. 
Since the moment sequence of the arcsine law characterizes the measure, convergence in law implies weak convergence.

\begin{thm}
As $c \to 0$, the discrete arcsine law $\mu_c$ weakly converges to the arcsine law $\dfrac{dx}{\pi \sqrt{2 -x^2}}$.
\end{thm}

If a measure on $\mathbb{R}$ has the same moment sequence as $\mu_c$,
it is identical to $\mu_c$.
\begin{thm}
The discrete arcsine law $\mu_c$ is characterized by its moments.
\end{thm}

\begin{proof}
We exploit the Carleman's condition for the moment sequence
\[\left\langle \left(\frac{e^{-it}}{\sqrt{2}} + \frac{e^{it}}{\sqrt{2}} + \dfrac{c}{i} \dfrac{d}{dt} \right)^m 1, 1 \right\rangle_{L^2(\{e^{it}\})}\]
of the discrete arcsine law. We may assume that $c > 0$, since $-c$ also gives the same moment sequence. Consider the Fourier expansion
\[\sum_{n} b_n^{(m)} e^{int} = \left(\frac{e^{-it}}{\sqrt{2}} + \frac{e^{it}}{\sqrt{2}} + \dfrac{c}{i} \dfrac{d}{dt} \right)^m 1.\]
Note that if $n \notin [-m, m]$ then $b_n^{(m)} = 0$.
By the equality
\[b_n^{(m + 1)} = \dfrac{b_{n - 1}^{(m)}}{\sqrt 2} + \dfrac{b_{n + 1}^{(m)}}{\sqrt 2} + \dfrac{c}{i} n b_n^{(k)},\]
we have 
\[\sum_{n=-m-1}^{m + 1} |b_n^{(m + 1)}| \le 
\left( \dfrac{1}{\sqrt{2}} + \dfrac{1}{\sqrt{2}} + c(m + 1) \right)^{m + 1}
\sum_{n=-m}^{m} |b_n^{(m)}| .\]
It is easy to show by induction that 
\[\sum_{n = - m}^m |b_n^{(m)}| \le \left( \sqrt{2} + cm \right)^m.\]
In particular the $(2m)$-th moment $b_0^{(2m)}$ is at most $(\sqrt{2} + 2cm)^{2m}$.
Therefore we have
\[\sum_{m= 0}^{\infty} \dfrac{1}{\sqrt[2m]{b_0^{(2m)}}} 
\ge \sum_{m= 0}^{\infty} \dfrac{1}{\sqrt{2} + 2cm} = +\infty.\]
This means that the moment sequence of the discrete arcsine law satisfies the Carleman's condition,
which is a sufficient condition for determinacy.  
For the Carleman's condition, we refer the readers to the book \cite{Akh} by Akhiezer.
\end{proof}

%
%
%

\section*{Acknowledgements}
The author is grateful to Prof. Marek Bo\.{z}ejko for his encouragements and suggestion to look at the work of Kerov \cite{KER}. He deeply 
thanks Prof. Izumi Ojima and Dr. Kazuya Okamura for many discussions on ``Quantum-Classical Correspondence''.

\end{document}